\documentclass{LMCS}

\def\dOi{10(3:3)2014}
\lmcsheading%
{\dOi}
{1--13}
{}
{}
{Sep.~12, 2013}
{Aug.~15, 2014}
{}

\ACMCCS{[{\bf Theory of computation}]: Computational complexity and cryptography; Logic}



\usepackage{amsmath, amssymb, amsthm}
\usepackage{xspace}
\usepackage{xcolor}
\usepackage{needspace}
\usepackage{hyperref}
\usepackage{ifthen}
\usepackage{fancyhdr}
\usepackage{tikz}
  \usetikzlibrary{arrows,automata}
 \usepackage{listings}
  
  \lstdefinelanguage{pseudo}{
    morekeywords={if,elseif,then,return,end,choose,guess,when,for,foreach,case},
    morekeywords=[3]{false,true,and,or,not},
    morecomment=[l]{//}
  }
\usepackage{microtype}

\newcommand{\comment}[2][]{{\color{blue}\scriptsize+++\ifthenelse{\equal{}{#1}}{}{#1:\ }#2---------}\marginpar{{\color{blue}$\bullet$}}}

\newcommand{\complexityClassFont}[1]{\mathrm{#1}} 
\newcommand{\logicClFont}[1]{\mathrm{#1}}        
\newcommand{\problemFont}[1]{\textsf{#1}}         
\newcommand{\mathCommandFont}[1]{\mathrm{#1}}     

\newcommand{\enc}[1]{\ensuremath\langle#1\rangle}

\newcommand{\CLASS}{\protect\ensuremath{\mathcal{C}}\xspace}

\newcommand{\PTIME}{\protect\ensuremath{\complexityClassFont{P}}\xspace}
\newcommand{\NP}{\protect\ensuremath{\complexityClassFont{NP}}\xspace}

\newcommand{\SO}{\protect\ensuremath{\logicClFont{SO}}\xspace}
\newcommand{\ESO}{\protect\ensuremath{\logicClFont{SO \exists}}\xspace}
\newcommand{\ESOHORN}{\protect\ensuremath{\logicClFont{SO \exists}}\textnormal{-Horn}\xspace}
\newcommand{\SOHORN}{\protect\ensuremath{\logicClFont{SO}\textnormal{-Horn}}\xspace}
\newcommand{\FO}{\protect\ensuremath{\logicClFont{FO}}\xspace}
\newcommand{\D}{\protect\ensuremath{\logicClFont{D}}\xspace}
\newcommand{\DHORN}{\protect\ensuremath{\logicClFont{D}\textnormal{-Horn}}\xspace}

\newcommand{\BD}{\protect\ensuremath{\logicClFont{D}^*}\xspace}
\newcommand{\BDHORN}{\protect\ensuremath{\BD\textnormal{-Horn}}\xspace}
\newcommand{\SDHORN}{\BDHORN}

\newcommand{\dep}[1][\cdot]{{\mathCommandFont{dep}\ifthenelse{\equal{#1}{}}{}{(\nobreak#1\nobreak)}}}

\newcommand{\LOGIC}{\ensuremath{\mathcal{L}}\xspace}
\newcommand{\MDL}[1][]{\ensuremath{\logicClFont{MDL}\ifthenelse{\equal{#1}{}}{}{(\allowbreak#1)}}\xspace}
\newcommand{\CTL}[1][]{\ensuremath{\logicClFont{CTL}\ifthenelse{\equal{#1}{}}{}{(#1)}}\xspace}
\newcommand{\LTL}[1][]{\ensuremath{\logicClFont{LTL}\ifthenelse{\equal{#1}{}}{}{(#1)}}\xspace}
\newcommand{\CTLs}[1][]{\ensuremath{\logicClFont{\CTL^*}\ifthenelse{\equal{#1}{}}{}{(#1)}}\xspace}

\newcommand{\MDLk}[1][]{\ensuremath{\logicClFont{MDL_k}\ifthenelse{\equal{#1}{}}{}{(\allowbreak#1)}}\xspace}

\newcommand{\DOMSET}{\protect\ensuremath{\problemFont{Dominating-Set}}\xspace}

\newcommand{\MDLSAT}[1][]{\ensuremath{\problemFont{MDL}\text{-}}\allowbreak\ensuremath{\problemFont{SAT}\ifthenelse{\equal{#1}{}}{}{(#1)}}\xspace}
\newcommand{\MDLMC}[1][]{\ensuremath{\problemFont{MDL}\text{-}}\allowbreak\ensuremath{\problemFont{MC}\ifthenelse{\equal{#1}{}}{}{(#1)}}\xspace}
\newcommand{\MDLMCparas}[3][]{\ensuremath{\problemFont{MDL\ifthenelse{\equal{#2}{}}{}{_{\mathnormal{#2}}}}\text{-}}\allowbreak\ensuremath{\problemFont{MC\ifthenelse{\equal{#3}{}}{}{_{\mathnormal{#3}}}}\ifthenelse{\equal{#1}{}}{}{(#1)}}\xspace}

\newcommand{\set}[3][]{\protect\ensuremath{\left\{#2\;\middle|\;\ifthenelse{\equal{#1}{}}{\text{#3}}{\parbox{#1}{#3}}\right\}}}

\newcommand{\dom}{{\rm Dom}} 
\newcommand{\Var}{{\rm Var}}
\newcommand{\Fr}{{\rm Fr}}
\newcommand{\mA}{{\mathfrak A}}
\newcommand{\rel}{\mathrm{rel}}

\theoremstyle{plain}
\theoremstyle{plain}
\theoremstyle{plain}
\theoremstyle{plain}

\begin{document}

	\title{A Fragment of Dependence Logic Capturing Polynomial Time}

	\author[J.~Ebbing]{Johannes Ebbing\rsuper a}	
	\address{{\lsuper{a,c,d}}Leibniz Universit\"at Hannover, Theoretical Computer Science, Appelstr.~4, 30167~Hannover, Germany.}	
	\email{\{ebbing,mueller,vollmer\}@thi.uni-hannover.de}  

	\author[J.~Kontinen]{Juha Kontinen\rsuper b}	
	\address{{\lsuper b}University of Helsinki, Department of Mathematics and Statistics, P.O. Box 68, 00014 Helsinki, Finland.}	
	\email{juha.kontinen@helsinki.fi}  

	\author[J.-S.~M\"uller]{Julian-Steffen M\"uller\rsuper c}	
	\address{\vspace{-18 pt}}	
	
	\author[H.~Vollmer]{Heribert Vollmer\rsuper d}	
	\address{\vspace{-18 pt}}	

	\thanks{{\lsuper{a,b,c,d}}This paper was supported by a grant from DAAD within the PPP programme, by DFG grant VO 630/6-2 and grants 264917 and 275241  of the Academy of Finland.}

\begin{abstract}
	In this paper we study the expressive power of Horn-formulae in dependence logic and show that they can express NP-complete problems.
	Therefore we define an even smaller fragment $\SDHORN$ and show that over finite successor structures it captures the complexity class \PTIME of all sets decidable in polynomial time. Furthermore, we show that 
  the open $\SDHORN$-formulae correspond to the negative fragment of $\ESOHORN$.

\end{abstract}

\maketitle

\section{Introduction}

Dependence logic, \D,  extends first-order logic by  dependence atomic formulae
\begin{equation}\label{da}\dep[t_1,\ldots,t_n]
\end{equation} the  meaning of which is that the value of the term $t_n$ is functionally determined by the values of $t_1,\ldots, t_{n-1}$. The semantics of $\D$ is defined in terms of sets of assignments (teams) instead of single assignments as in first-order logic.
While in first-order logic the order of quantifiers solely determines the dependence relations between variables, in dependence logic more general dependencies between variables can be expressed. Historically dependence logic was preceded by  partially ordered quantifiers (Henkin quantifiers) of Henkin  \cite{he61}, and Independence-Friendly (IF) logic of Hintikka and Sandu \cite{hisa89}. It is known that both IF logic and dependence logic are equivalent to existential second-order logic $\ESO$ in expressive power.
From the point of view of descriptive complexity theory, this means that dependence logic captures the class $\NP$.

The framework of dependence logic has turned out be flexible to allow  
interesting generalizations. For example, the extensions of dependence logic in terms of so-called intuitionistic implication and linear implication was  introduced in \cite{AbVa}. In \cite{yaar} it was shown that extending $\D$ by the intuitionistic implication makes the logic equivalent to full second-order logic $\SO$. On the other hand, in  \cite{arebkovo11}
the extension of \D by a majority quantifier was defined and shown to capture the Counting Hierarchy. Furthermore, new variants of the dependence atomic formulae have also  been introduced in \cite{vagr10}, \cite{Galliani12},  and \cite{en11}. 

In this paper we study certain  fragments of dependence logic. While it is known that \D captures the class $\NP$, the complexity of various syntactic fragments of \D are not yet fully understood.  Some work has been done in this direction:
\begin{itemize}

\item All sentences of \D  can be transformed to a form  
\begin{equation}\label{D-NF}
		\forall\overline{x}\exists\overline{y}(\bigwedge_i \dep[\overline{z}_i,w_i] \wedge \psi),
 \end{equation}
where $\psi$ is quantifier-free first-order formula \cite{va07}.

\item In formulae of form (\ref{D-NF}), the use of those variables depending on others can even further be restricted; in a sense, only Boolean information in form of equality tests is needed. We will introduce this fragment \BD formally and show that it is as expressive as \D.

\item  The fragments of \D defined either by restricting the number of universal quantifiers or the arity of  dependence atoms in sentences were mapped to the corresponding sublogics of \ESO in \cite{ADJK}. Making use of the well-known time hierarchy theorem this implies a strict hierarchy of fragments within \D.

\item The existential sentences of \D collapse to \FO \cite{{ADJK}}, whereas the universal sentences can define NP-complete problems. 
\end{itemize}

\noindent The last remark above follows from the result of \cite{2010Jarmo} 
showing that  the question of deciding whether a team $X$ satisfies $\phi$, where  
\[ \phi= \dep[x,y]\vee\dep[u,v]\vee \dep[u,v]\]
is $\NP$-complete and from  the observation that $\phi$ can be translated to an equivalent  universal sentence of \D (see the proof of Lemma \ref{lemma1}). 

In this paper our main objects of study are Horn fragments of \D. In analogy to (\ref{D-NF}) we first define \DHORN to be the set of formulae of the form:
\[
		\phi = \forall\overline{x}\exists\overline{y}(\bigwedge_i \dep[\overline{z}_i,w_i] \wedge \bigwedge_j C_j),
	\]
where the $C_j$ are clauses, i.e., disjunctions of atomic and negated atomic \FO-formulae, that contain at most one positive formula with an occurrence of an existentially quantified variable.
While we will show that this fragment essentially is as expressive as full dependence logic (i.e., it can express NP-complete problems), 
we will prove that a slightly more restricted fragment, denoted by \SDHORN and obtained from \DHORN in exactly the same way as \BD is obtained from \D,
corresponds over finite successor structures  to the class of second-order Horn formulae which, by a famous result by Gr\"adel \cite{gr92}, are known to capture \PTIME over finite successor structures. 
The result of \cite{gr92} thus allows us to conclude that the sentences of \SDHORN  also capture \PTIME. 
An interesting question is whether $\SDHORN$ can be somehow extended to approach the major open question of descriptive complexity, whether there is a logic for \PTIME properties of structures in the absence of a built-in ordering relation. 
We also consider the complexity of \SDHORN formulae with free variables and show that, over successor structures, the open $\SDHORN$-formulae correspond exactly to the negative fragment of $\ESOHORN$. This result is analogous to a result of \cite{kova09}, who showed that 
the open $\D$-formulae correspond exactly to the downward closed NP-properties.

We would like to point out that very recently, independent work of Galliani and Hella \cite{gahe13} obtained a characterization of \PTIME in terms of inclusion logic, a variant of dependence logic where instead of dependence atoms, so called inclusion atoms are used.





This article is organized as follows. In the next section, we introduce dependence logic and some of its basic properties. We also recall Gr\"adel's characterization of \PTIME in terms of second-order Horn logic. In Sect.~\ref{sect:DHorn} we define our fragments of dependence logic, the Horn fragment and the strict Horn fragment. In Sect.~\ref{sect:results} we present our characterization of \PTIME, and in  Sect.~\ref{case_open_formulas}, we consider the open formulae of \SDHORN.

\section{Preliminaries}\label{sect:prelim}


\subsection{Dependence Logic}

In this section we will define the semantics of dependence logic. 
Dependence logic (\D) extends first order logic by new atomic formulae expressing dependencies between variables. 

\begin{defi}[\cite{va07}]
Let $\tau$ be a vocabulary. The set of $\tau$-formulae of dependence logic ($\D[\tau]$) is defined by extending the set of $\tau$-formulae of first order logic ($\FO[\tau]$) by dependence atoms of the form
\begin{align}\label{dep_atom}
\dep[t_1,\dots,t_n],
\end{align}
where $t_1,\dots,t_n$ are terms. 
\end{defi}

In this paper, we only consider formulae in negation normal form; this means that negation occurs only in front of atomic formulae. 

\begin{defi}
Let $\varphi$ be a dependence logic formula. We define the set $\Fr(\varphi)$ of free variables occurring in $\varphi$ as in first order logic with the additional rule
\[
\Fr(\dep[t_1,\dots,t_n]) = \Var(t_1) \cup \dots \cup \Var(t_n),
\]
where $\Var(t_i)$ is the set of variables occurring in $t_i$. A formula $\varphi$ with $\Var(\varphi) = \emptyset$ is called a sentence.
\end{defi}

Now we define team semantics for dependence logic. Satisfaction for dependence logic formulae  will be defined with respect to  \emph{teams} which are \emph{sets of assignments}. Formally, teams are defined as follows.

\begin{defi}
Let $A$ be a set and $\{x_1,\dots,x_n\}$ be a set of variables. 
\begin{itemize}
\item Then a \emph{team} $X$ over $A$ is a set of assignments $s\colon \{x_1,\dots,x_n\} \to A$. We refer to $\{x_1,\dots,x_n\}$ as the \emph{domain} and to $A$ as the \emph{co-domain} of $X$.
\item The relation $\rel(X)$ over $A^n$ corresponding to $X$ is defined as follows
\[
\rel(X) = \{(s(x_1),\dots,s(x_n))\mid s \in X\}
\]
\item Let $F\colon X \to A$ be a function, then we define
\begin{eqnarray*}
X(F/x)&=&\{s(F(s)/x) : s\in X \}\\ 
 X(A/x)&=&\{s (a/x): s\in X\ \textrm{and}\ a\in A \}.
\end{eqnarray*}
\end{itemize}
\end{defi}\medskip

\noindent We are now able to define team semantics. In the following definition, $t^\mA\langle s\rangle$ for a term $t$ and an assignment $s$ denotes the value of $t$ under $s$ in structure $\mA$.

\begin{defi}(\cite{va07})\label{def:semantics}
Let $\mA$ be a model and $X$ a team of $A$. Then we define the relation $\mA \models_{X} \varphi$ as follows:
\begin{itemize}
\item If $\varphi$ is a first-order literal, then $\mA\models _X\varphi$ iff for all $s\in X$ we have  $\mA\models_s\varphi$, where $\models_{s}$ refers to satisfaction in first-order logic.
\item $\mA\models _X \dep[t_{1},\ldots,t_{n}]$ iff for all $s,s'\in
X$ such that $t_1^{\mA}\langle s\rangle  =t_1^{\mA}\langle
s'\rangle  ,\ldots, t_{n-1}^{\mA}\langle s\rangle
=t_{n-1}^{\mA}\langle s'\rangle  $, we have $t_n^{\mA}\langle
s\rangle  =t_n^{\mA}\langle s'\rangle  $.
\item  $\mA \models _X \neg \dep[t_{1},\ldots,t_{n}]$ iff
$X=\emptyset$.
\item $\mA\models _X \psi \wedge \varphi$ iff $\mA\models _X \psi$ and $\mA\models _X \varphi$.
\item $\mA\models _X \psi \vee \varphi$ iff $X=Y\cup Z$ such that $\mA\models _Y \psi$  and $\mA\models _Z \varphi$ .

\item   $\mA \models _X \exists x\psi$ iff $\mA \models _{X(F/x)} \psi$ for some $F\colon X\to A$.

\item $\mA \models _X \forall x\psi$ iff $\mA \models _{X(A/x)} \psi$.
\end{itemize}
Above, we assume that the domain of $X$ contains the variables free in $\varphi$. Finally, a sentence $\varphi$ is true in a model $\mA$  (in symbols: $\mA\models \varphi$)  if $\mA\models _{\{\emptyset\}} \varphi$.
\end{defi}

Let us then  recall some basic properties of dependence logic that will be needed later. The following result shows that
the truth of a \D-formula depends only on the interpretations of variables occurring free in the formula. Below,  for $V\subseteq \dom(X)$, $X\upharpoonright V$ is defined by
 \[X\upharpoonright V = \{s\upharpoonright V \mid s\in X\}.\]
\begin{thm}[\cite{va07}]\label{freevar}
Suppose $V\supseteq \Fr(\phi)$. Then $\mA \models _X\phi$ if and only if $\mA \models _{X\upharpoonright V} \phi$.
\end{thm}

All formulae of dependence logic also satisfy the following strong monotonicity property called  \emph{Downward Closure}.
\begin{thm}[\cite{va07}]\label{Downward closure}Let $\phi$ be a formula of dependence logic, $\mA$  a model, and $Y\subseteq X$ teams. Then $\mA\models_X \phi$ implies $\mA\models_Y\phi$.
\end{thm}

Finally, we note that  dependence logic is a conservative extension of first-order logic.
\begin{defi} A formula $\phi$ of \D is called a first-order formula if it does not contain dependence atomic formulae as subformulae.  \end{defi}

First-order formulae of dependence logic satisfies the so-called \textit{flatness} property:
\begin{thm}[\cite{va07}]\label{FO} Let $\phi$ be a first-order formula of dependence logic. Then for all $\mA$ and $X$ it holds that
$\mA\models _X\phi \textrm{ if and only if for all }\allowbreak s\in X \textrm{ we have } \mA\models_s\phi$.
\end{thm}

In order to study the expressive power of a logic in terms of computational complexity, we have to define the notion of ``capturing'' of a complexity class. 


\begin{defi}
	Let $\mathcal{O}$ be the class of all finite \emph{successor structures}, i.e., all structures $\mA$ with $|\mA|= \{0,\dots,n-1\}$ that contain, possibly among other constants and relations, the constants $0$ and $\max$ and the relation $R = \{(x,x+1)\mid x \leq n-2\}$ for some $n \geq 1$. 
\end{defi}

\begin{defi}\hfill
\begin{enumerate}
\item Let $\LOGIC_{1}$ and $\LOGIC_{2}$ be two logics. Then $\LOGIC_{1}\equiv\LOGIC_{2}$ denotes that for every $\LOGIC_{1}$ sentence there is an equivalent $\LOGIC_{2}$ sentence, and vice-versa.
\item Let \CLASS be complexity class and \LOGIC be a logic. Then \CLASS is captured by \LOGIC, in symbols: $\CLASS\equiv\LOGIC$, if for all problems $P \subseteq\mathcal{O}$ it holds that $P \in \CLASS$ if and only if there exists an \LOGIC sentence $\phi$ for which $P = \{\mA\in\mathcal{O} \mid \mA \models \phi\}$ holds. 
\end{enumerate}
\end{defi}\medskip

\noindent For dependence logic it is known that it captures \NP, since $\D\equiv\Sigma^1_{1}$ \cite{va07} and $\Sigma^1_{1}\equiv\NP$ \cite{fag74}.

\begin{thm}\label{D=NP}
$\D \equiv \NP$.
\end{thm}

As said, one cornerstone in the proof of the preceding theorem is a transformation of $\Sigma^1_{1}$-formulae into $\D$-formulae. It can be observed that the appearing dependencies are of a very particular form.

\begin{defi}[\BD]\label{def:bd}
	A Boolean dependence formula (\BD formula) is a \D formula of the form
	\[
		\phi = \forall\overline{x}\exists\overline{y}\left(\bigwedge_{1\le i\le n} \dep[\overline{z}_i,y_i] \wedge\theta\right),
	\]
	where $\overline{y}=(y_1,\ldots, y_n)$ are pairwise distinct variables, $\overline{z}_i\subseteq \overline{x}$, and $\theta$ is an arbitrary first-order formula, but 
	the existentially quantified variables can only appear in atomic formulae of the form $y_i = 0$ or $y_i = y_j$.
\end{defi}

The next result shows that $\BD$ is essentially as expressive as full dependence logic, even in the absence of a successor relation.

\begin{thm}\label{thm:BD=D}
\begin{enumerate}
\item Over  structures with built-in plus and times, $\BD\equiv\D$. 
\item The logic \BD can define \NP-complete problems even without the built-in successor relation.
\end{enumerate}
\end{thm}

\begin{proof}
It is known that \NP can be captured by the class of second-order formulae of the form
$\phi = \exists{P_1}\dots\exists{P_n}\forall\overline{x}\theta$, where $\theta$ is a quantifier-free FO-formula, over structures with built-in predicates for plus and times \cite{im89}.
As in V\"a\"an\"anen's proof we transform this into a \D formula, see \cite[Sect.~6.3]{va07}. It can be observed that the variables that appear as last variable in a dependence atom are only used in the form given in Def.~\ref{def:bd} above. This proves the first claim.

The second claim follows from the fact that the fragment Strict \NP, SNP \cite{kova87}, 
of sentences $\exists{P_1}\dots\exists{P_n}\forall\overline{x}\theta$,
where $\theta$ is quantifier free without order, translates into $\BD$ under the same translation as used above, but SNP can define NP-complete problems.
\end{proof}


\subsection{Horn Logic}

In this section we introduce first and second-order Horn logic and discuss their expressive power.

\begin{defi}[\FO-Horn]
A \emph{clause} is a disjunction of atomic and negated atomic formulae, including $\bot$ and $\top$. A \emph{Horn clause} is a clause with at most one non-negated atom.
\end{defi}

\begin{defi}[\SO-Horn, \cite{gr92}]\label{Gradel_HORN}
A second-order Horn formula is a second-order formula $\phi$ of the form
\[
Q_1 P_1 \dots Q_k P_k \forall \overline x \bigwedge_j C_j,
\]
where for all $i,j$, $Q_i \in \{\forall, \exists\}$, $P_i$ are relation symbols, $C_{j}$ are clauses that contain at most one positive occurrence of a predicate $P_{i}$.
We denote by \SO-Horn the set of all second-order Horn formulae.
The existential fragment of \SO-Horn, denoted \ESOHORN, is the fragment where all $Q_i = \exists$.
\end{defi}

\begin{thm}[\cite{gr92}]\label{ESOHORN}
\begin{enumerate}
\item $\SOHORN\equiv\ESOHORN$.
\item $\ESOHORN \equiv \PTIME$.
\end{enumerate}
\end{thm}

\section{Horn Fragments of Dependence Logic}\label{sect:DHorn}

The Horn requirement that clauses contain at most one non-negative atom can also be transferred into the context of dependence logic. 

\begin{defi}[\DHORN]\label{def_dhorn}
	A \DHORN formula is a \D formula of the form
	\[
		\phi = \forall\overline{x}\exists\overline{y}(\bigwedge_{1\le i\le n} \dep[\overline{z}_i,y_i] \wedge \bigwedge_j C_j),
	\]
where $\overline{y}=(y_1,\ldots, y_n)$ are pairwise distinct variables,
 $\overline{z}_i\subseteq\overline{x}$, and each clause $C_j$ contains at most one positive atomic formula with an occurrence of an existentially quantified variable.
\end{defi}

As we will later show, the \DHORN fragment is essentially as expressive as full dependence logic, as it can express \NP-complete problems, even without successor. But we will prove that a subfragment that informally uses only binary information about the dependent variables,
exactly in the same way as in the logic \BD introduced above,
has a possibly lower expressive power, since it corresponds to polynomial time computation.

\begin{defi}[\SDHORN]
	A Boolean \DHORN formula (\SDHORN formula) is a \DHORN formula
	\[
		\phi = \forall\overline{x}\exists\overline{y}(\bigwedge_i \dep[\overline{z}_i,y_i] \wedge \bigwedge_j C_j),
	\]
with the additional condition that the existentially quantified variables can only appear in atomic formulae of the form $y_i = 0$ or $y_i = y_j$.
\end{defi}


\section{Expressive power of the Horn fragments}\label{sect:results}

\subsection{\SDHORN captures \PTIME}

In this section we investigate the expressive power of the fragments \DHORN and \SDHORN. We start by proving that the latter logic captures polynomial time.

\begin{thm}\label{SDHORN=PTIME}
	The logic \SDHORN captures the complexity class \PTIME on finite successor structures: $\PTIME \equiv \SDHORN$.
\end{thm}

We prove the theorem with the two following lemmas.

%
	
	\begin{lem}\label{SOHORN_in_SDHORN}
		$\SOHORN\subseteq\SDHORN$ on finite successor structures.	
	\end{lem}
	\begin{proof}By Theorem  \ref{ESOHORN}, it suffices to show the claim for \ESOHORN.	Let 
\begin{equation*}
\phi = \exists{P_1}\dots\exists{P_n}\forall\overline{x}\bigwedge_{j}C_j
\end{equation*}
be a \ESOHORN sentence. First we  replace quantification over relations by quantification over functions in $\phi$. This is achieved by coding each relation $P_i$ by its characteristic function. Furthermore, the clauses $C_j$ of $\phi$ are replaced by clauses  $C'_j$  acquired by replacing the  atomic formulae  $P_i(\overline{z})$  by $F_i(\overline{z})=0$. After these transformations, we get a sentence $\phi'$ which is logically equivalent to $\phi$, where
					\[\phi' = \exists{F_1}\dots\exists{F_n}\forall\overline{x}\bigwedge_{j}C'_j.\]		
In dependence logic, the functions $F_i$ will be translated by existentially quantified variables.  For this reason, each occurrence of $F_i$ in $\phi'$ has to be of the form $F_i(\overline{z}_i)$ for some unique tuple $\overline{z}_i$ of pairwise distinct variables.
This can be accomplished  analogously to Theorem 6.15 in \cite{va07}.   First, one by one, we replace  each occurrence of  every term  $F_i(\overline{t})$, where $t=(t_1,\ldots,t_k)$, in $\phi'$ by a new term $F_i(\overline{w})$, where $\overline{w}=(w_1,\ldots,w_k)$ is a fresh tuple of pairwise distinct variables quantified universally, and use the equivalence of $\bigwedge_{j}C'_j(F_i(\overline{t}))$ and 
\begin{equation}\label{terms}
\forall \overline{w}( \bigwedge_{1\le p\le k}w_p=t_p\rightarrow \bigwedge_{j}C'_j[ F_i(\overline{w})/F_i(\overline{t})] ).  
\end{equation}
In this way, $\bigwedge_{j}C'_j$ is transformed to an equivalent formula  containing only simple terms of the form $F_i(\overline{w})$. It is easy to 
see that using distributivity laws of \FO, the quantifier-free part of formula \eqref{terms} can be translated to a conjunction of clauses satisfying the condition of Definition \ref{Gradel_HORN}, i.e., these transformations do not carry us outside of \ESOHORN. 

Let us now assume that all occurences of the symbols $F_i$ are of the form $F_i(\overline{w})$ for some tuple $\overline{w}$  of pairwise distinct variables. We still need to ensure that for each $F_i$ the tuple  $\overline{w}$ is unique. Suppose that this in not the case, i.e., assume that  $\phi'$ contains two occurrences $F_i(\overline{z})$ and  $F_i(\overline{z}')$
of  the same $F_i$  but with  different tuples of variables. Now the idea is to replace all  occurrences of $F_i(\overline{z}')$  by $G(\overline{z}')$, where $G$ is fresh symbol,  and use the fact that $\phi'$ is equivalent to:
		
		\begin{equation}\label{claim-so-d:dup}	\exists{F_1}\dots\exists{F_n}\exists{G}\forall\overline{x}\bigwedge_{j}C_j'[G(\overline{z}')/F_i(\overline{z}')]
			\wedge (\bigwedge _kz_k = z'_k \rightarrow F_i(\overline{z})=G(\overline{z}'))
		\end{equation}
Again note that this transformation does not carry us outside of  \ESOHORN. Hence we may assume that in  $\phi'$  all occurrences of $F_i$ are of the form $F_i(\overline{z}_i)$ for some unique tuple $\overline{z}_i$ of pairwise distinct variables. We  apply the result of \cite[Theorem 6.15]{va07}  which allows us to directly translate $\phi'$ into $\SDHORN$ which looks as follows: 
		
		\begin{equation}\label{claim-so-d:dup2}
			 \forall\overline{x}\exists y_1\dots\exists y_n (\bigwedge_{i=1}^n \dep[\overline{z}_i,y_i] \wedge \bigwedge_{j}C_j'[y_i/ F_i(\overline{z}_i)]).
		\end{equation}
	\end{proof}

\begin{lem}\label{SDHORN_in_ESOHORN}
	$\SDHORN\subseteq \ESOHORN$ on finite successor structures.
\end{lem}

	\begin{proof}
	We prove the lemma by showing that every $\SDHORN$ sentence can be translated to an equivalent \ESOHORN sentence. Let 
	 \[\phi = \forall\overline{x}\exists y_1\ldots \exists y_k(\bigwedge_{i=1}^{h}\dep[\overline{z}_i,y_i] \wedge \bigwedge_{j}C_j)\]
	 be a $\SDHORN$ sentence.  Without loss of generality, we assume that atoms $y_j = y_j$ do not appear in $\phi$.
	By the results of \cite{va07}, this sentence is equivalent to the \ESO sentence $\phi^*$
	\[    \exists f_1\ldots \exists f_k \forall\overline{x}   \bigwedge_{j}C^*_j,      \] 
	where $C^*_j$ is defined by replacing the occurrences of $y_i$ by the term $f_i(\overline{z}_i)$.
	It now suffices to transform this sentence to an equivalent  \ESOHORN sentence. This is achieved as follows. For each pair $\{r,s\}$ such that either $f_r(\overline{z}_r)= f_s(\overline{z}_s)$
	 or  $f_s(\overline{z}_s)= f_r(\overline{z}_r)$ occurs in $\phi^*$ we  introduce a new relation symbol $P_{\{ r,s\}}$. 
	Furthermore, for each $\{r,0\}$ such that $f_r(\overline{z}_r)=0$ occurs in  $\phi^*$  a relation symbol  $P_{\{r,0\}}$  is introduced. 
	Finally for each pair of relation symbols $P_{\{r,s\}}$, $P_{\{s,t\}}$ we add the new relation symbol $P_{\{r,t\}}$.  All these new relation symbols have arity $|\overline{x}|$. In the following $s,r$ and $t$  range over $\{0,1,\ldots,k\}$ and atoms of the form $f_r(\overline{z}_r)=0$ are  also written as $f_r(\overline{z}_s)= f_0(\overline{z}_0)$.

Now for each $C^*_j$ in $\phi^*$, the formula $C'_j$ is defined by replacing the atoms $f_r(\overline{z}_r)=f_s(\overline{z}_s)$ and $f_r(\overline{z}_r)=0$ by the atoms  $P_{\{r,s\}}(\overline{x})$ and  $P_{\{r,0\}}(\overline{x})$, respectively. Since $\phi\in \SDHORN$ it follows immediately that 
	\[  \phi' :=  \exists \overline{P} \forall\overline{x} \bigwedge_{j}C'_j   \] 
	is a  \ESOHORN sentence. In order to guarantee the equivalence of $\phi^*$ and $\phi'$ (and hence the equivalence $\phi' $ and $\phi$) we add to $\bigwedge_{j}C'_j$ the following clauses:
\begin{equation}\label{ident}
 \bigwedge _{r,s,t} [(P_{\{r,s\}}(\overline{x})\wedge   P_{\{s,t\}}(\overline{x}))\rightarrow   P_{\{r,t\}}(\overline{x})],   
 \end{equation}
axiomatizing the transitivity of identity. We also replace the quantifier prefix  $\forall \overline{x}$ by $\forall \overline{x} \forall \overline{x}' $ and add the following clauses to the formula: 
\begin{equation}\label{funct2}
	 \bigwedge _{r} \bigwedge _{s} ((\overline{z}_r =\overline{z}_r'\rightarrow ( P_{\{r,s\}}(\overline{x})\leftrightarrow P_{\{r,s\}}(\overline{x}'))),  
	 \end{equation}
which ensures that the identities satisfied by the term $f_r(\overline{z}_r)$ are determined by the values of the variables $\overline{z}_r$. It is straightforward to check that $\phi'$, with the modifications in \eqref{ident} and \eqref{funct2}, is equivalent to a formula in \ESOHORN. We will next show that, for structures $\mA$ with $|A|\ge k+1$, where $k$ is the number of functions $f_i$ in $\phi^*$, the sentences  $\phi'$ and  $\phi^*$ are equivalent.  Note that, by modifying the behaviour of $\phi'$ in the finitely many structures of cardinality at most $k$, we can  find a sentence of  \ESOHORN that is logically equivalent to  $\phi^*$, and hence $\phi$.  

The proof of the implication from $\mA\models \phi^*$ to 
 $\mA\models \phi'$ is straightforward, hence we consider only the converse implication. Suppose then that $\mA$ is a structure with $|A|\ge k+1$, and $\mA\models \phi'$. We need to show that $\mA\models  \phi^*$. It suffices to find interpretations $g_i\colon A^{|\overline{z}_i|}\rightarrow A$ for the function symbols $f_i$, for $1\le i \le k$, such that   
\[(\mA,g_1,\ldots,g_k)\models \forall\overline{x}   \bigwedge_{j}C^*_j.\]
Since $\mA\models \phi'$, there are relations $S_{\{s,r\}}\subseteq A^{|\overline{x}|}$  such that 
\[(\mA,(S_{\{s,r\}})_{s,r})\models \forall \overline{x} \forall \overline{x}'  \bigwedge_{j}C'_j.\]
We define the  functions $g_i$ as follows:
$g_{i}(\overline{a})=l$, if $l$ is the smallest integer in $\{0,\ldots, i-1\}$ such that $\overline{a}$ can be extended to a tuple $\overline{a}'\in A^{|\overline{x}|}$ such that $\overline{a}'\in S_{\{i,l\}}$. Otherwise, $g_{i}(\overline{a})=i$.  The formula \eqref{funct2} ensures that $g_{i}$ is well-defined. Now, using the fact that the relations  $(S_{\{s,r\}})_{s,r}$ satisfy the formula \eqref{ident}, it is straightforward to  show that for all $\overline{a}\in A^{|\overline{x}|}$:
\begin{equation}\label{eqfin}
  (\mA,g_1,\ldots,g_k)\models_v f_s(\overline{z}_s)=f_r(\overline{z}_r) \Leftrightarrow 
(\mA,(S_{\{s,r\}})_{s,r} )\models_v   P_{\{s,r\}}(\overline{x}), 
\end{equation}
where $v$ is the assignment such that $v(\overline{x})=\overline{a}$. Using \eqref{eqfin},  $\mA\models \phi^*$ easily follows.  \end{proof}

\begin{cor}\label{k_variables}
$\SDHORN$ restricted to formulae with $k$ universal quantifiers captures a subclass of $\mathrm{TIME}(n^{2k})$.
\end{cor}

\begin{proof}
Starting with a formula $\psi\in\SDHORN$ with $k$ universal quantifiers, we obtain a formula $\phi\in\ESOHORN$ with $\ell=2k$ universal quantifiers as in the proof of the preceding lemma.
As shown by \cite[Corollary 4.2]{gr92}, a formula $\phi$ from \ESOHORN can be transformed into an equivalent propositional Horn formula $\phi'$, such that the obtained formula has size $n^\ell$, where $n$ denotes the size of the model and $\ell$ the number of universal quantifiers in $\phi$. 
The claim now follows since evaluation of propositional Horn formulae is in linear time.
\end{proof}

It is worth noting that Lemma \ref{SDHORN_in_ESOHORN} does not hold without the built-in successor relation, e.g.,  the following $\SDHORN$-sentence
\[  \exists y_1\ldots \exists y_n \bigwedge _{1\le i<j\le n} \neg y_i=y_j ,\]
is not logically equivalent to any \ESOHORN-sentence since the properties definable in \ESOHORN without successor are closed under taking substructures  \cite{gr92}.

\subsection{\DHORN expresses \NP-complete problems}

In the rest of this section we turn to the logic \DHORN. We show that it is essentially as expressive as \D by proving that the following \NP-complete problem \DOMSET can be expressed in \DHORN. 
\[
\DOMSET = \set[7.5cm]{\enc{(V,E), k}}{there is a set $V' \subseteq V$, $|V'| \leq k$,  such that for every $v\in V\backslash V'$ there is a $u \in V'$ with $(u,v) \in E$}.
\]

\begin{thm}
$\DOMSET$ can be defined in \DHORN, even without successor.
\end{thm}

\begin{proof}
Consider the following formula
\begin{equation}\label{DS-Formel}
\phi = \begin{array}[t]{@{}r@{}l}
\forall x_{0}\forall x_{1}\forall x_{2}
\exists y_{0}\exists y_{1}\exists y_{2}
\bigl( &
\dep[x_{0},y_{0}] \\ &
{}\wedge \dep[x_{1},y_{1}] \wedge \dep[x_{2},y_{2}] \\ &
{}\wedge (x_{1}=x_{2}\rightarrow y_{1}=y_{2}) \wedge (y_{1}=y_{2}\rightarrow x_{1}=x_{2}) \\ &
{}\wedge E(x_{0},y_{0}) \wedge (y_{0}=x_{1}\rightarrow P(y_{1})
\bigr).
\end{array}
\end{equation}
We claim that, for $G=(V,E)$,  $\langle G,k\rangle \in\DOMSET \iff \langle G^*,k\rangle \in \DOMSET \allowbreak \iff (G^*,P)\models\phi$, where $P$ is an arbitrary unary relation that contains $k$ nodes and $G^*$ extends $G$ by self loops.

To see this, note that the atoms in the second and third line in (\ref{DS-Formel}) define a bijection $f$ under which $x_{1}$ is mapped to $y_{1}$ (this is analogous to the example on p.~51 in \cite{va07}). The rest of the formula ensures that every node $x_{0}$ is connected by an edge to some node $y_{0}$ and $y_{0}$ is in bijection with some element in $P$. Since there are only $k$ such elements, we express existence of a dominating set of cardinality $k$.
\end{proof}


%
%
%

\section{The case of open formulae}\label{case_open_formulas}
In this section we show that over successor structures the open $\SDHORN$-formulae correspond exactly to the negative fragment of $\ESOHORN$. 


Theorem~\ref{D=NP} shows that sentences of \D correspond to sentences of $\ESO$. Note that this result does not tell us anything about formulae of dependence logic with free variables. An upper bound for the complexity of formulae of \D is provided by the following result.
\begin{thm}[\cite{va07}]\label{DtoSIG}
Let $\tau$ be a vocabulary and $\varphi$ a  $\D[\tau]$-formula  with free variables $x_1,\dotsc, x_k$. Then
there is a $\tau\cup \{R\}$-sentence $\psi$ of $\ESO$, in which  $R$  appears only negatively,  such that for all models $\mA$ and teams $X$ with domain $\{x_1,\dotsc, x_k\}$:
\[ \mA\models_X \phi \iff (\mA,\rel(X))\models\psi . \]
\end{thm}

In \cite{kova09} it was shown that also the converse holds.

\begin{thm}[\cite{kova09}]\label{Dformulas} Let $\tau$ be a vocabulary and  $R$ a  $k$-ary relation symbol such that   $R\notin \tau$. Then for every  $\tau\cup \{R\}$-sentence $\psi$ of $\ESO$, in which $R$ appears only negatively,  there is a
$\tau$-formula $\phi$ of \D  with free variables $x_1,\dotsc, x_k$  such that,  for all  $\mA$ and  $X\neq \emptyset$  with domain $\{x_1,\dotsc, x_k\}$:
\begin{equation}\label{new_eq}
 \mA\models_X \phi \iff (\mA,\rel(X))\models\psi .
\end{equation}
\end{thm}
Theorem \ref{Dformulas} shows that formulae of dependence logic correspond in a precise way to the negative fragment of $\ESO$ and are therefore very expressive. We will next proceed to generalize Theorems \ref{DtoSIG} and \ref{Dformulas} to  open $\SDHORN$ formulae and the negative fragment of $\ESOHORN$. We consider first the analogue of Theorem \ref{DtoSIG}. 


\begin{lem}\label{lemma1} Let $\tau$ be a vocabulary and $\varphi\in \SDHORN[\tau]$  with free variables $z_1,\dotsc, z_k$. Then
there is a $\tau\cup \{R\}$-sentence $\psi$ of $\SDHORN$ in which $R$  appears only negatively, such that for all models $\mA$ and teams $X$ with domain $\{z_1,\dotsc, z_k\}$:
\begin{equation}\label{rep}
\mA\models_X \phi \iff (\mA,\rel(X))\models\psi . 
\end{equation}
\end{lem}
\begin{proof} Suppose that $\varphi$ is of the form
\[
		\forall\overline{x}\exists\overline{y}(\bigwedge_i \dep[\overline{w}_i,y_i] \wedge \bigwedge_j C_j).
	\]
As generally showed in Proposition 5.4 in \cite{ADJK}, $\varphi$ can be translated to a \D sentence $\psi'$ satisfying the equivalence  \eqref{rep}: 
\[ \psi':= \forall \overline{z} (\neg R(\overline{z})\vee \varphi). \]
Therefore it suffices to show that $\psi'$ is equivalent to some $\SDHORN$ sentence $\psi$. We may assume that the variables in $\overline{x}$ and $\overline{y}$ do not appear in $\neg R(\overline{z})$, hence $\psi'$ is equivalent to 
\begin{equation}\label{formula1}
		\forall \overline{z}\overline{x}\exists\overline{y}((\bigwedge_i \dep[\overline{w}_i,y_i] \wedge \bigwedge_j C_j)\vee \neg R(\overline{z})).
\end{equation}
The proof of Lemma 3.2 in \cite{ADJK} shows that the following subformula of \eqref{formula1} 
\[ \exists\overline{y}((\bigwedge_i \dep[\overline{w}_i,y_i] \wedge \bigwedge_j C_j)\vee \neg R(\overline{z}))
	\]
is equivalent to 
\[ \exists\overline{y}(\bigwedge_i \dep[\overline{w}_i,y_i] \wedge ( \bigwedge_j C_j\vee \neg R(\overline{z}))).
	\]
Therefore, the sentence $\psi'$ is logically equivalent to the sentence
\begin{equation}\label{formula2}
 \forall \overline{z}\overline{x}\exists\overline{y}(\bigwedge_i \dep[\overline{w}_i,y_i] \wedge ( \bigwedge_j C_j\vee \neg R(\overline{z})))
	.
\end{equation}
Note that since $( \bigwedge_j C_j\vee \neg R(\overline{z}))$ is first-order, Theorem \ref{FO} implies  that  \eqref{formula2} is equivalent to the  $\SDHORN$ sentence $\psi$
\[ \psi= \forall \overline{z}\overline{x}\exists\overline{y}(\bigwedge_i \dep[\overline{w}_i,y_i] \wedge \bigwedge_j C'_j),\]
where $C'_j $ is $C_j\vee \neg R(\overline{z})$.
\end{proof}
By combining Lemma \ref{lemma1} and our translation from $\SDHORN$ to $\ESOHORN$  the following analogue of Theorem \ref{DtoSIG} follows.

\begin{thm}
Let $\tau $ be a vocabulary such that  $\{0,\max,R \}\subseteq \tau$, and $\varphi$ a $\SDHORN[\tau]$-formula with free variables $z_1,\dotsc, z_k$. Then
there is a $\tau\cup \{R\}$-sentence $\psi$ of $\ESOHORN$, in which $R$  appears only negatively, such that for all sufficiently large models $\mA$ and teams $X\neq \emptyset$ with domain $\{z_1,\dotsc, z_k\}$:
\[ \mA\models_X \phi \iff (\mA,\rel(X))\models\psi . \]
\end{thm}
\begin{proof} 
Note that if $R$ appears only negatively in  $\phi$, then it will also only appear negatively in $\phi'$, as defined in  Lemma \ref{SDHORN_in_ESOHORN}, and these sentences are equivalent for large enough structures. 
\end{proof}

Next we show that the analogue of Theorem \ref{Dformulas} also holds.

\begin{thm}\label{SDHORN-formulas} Let $\tau$ be a vocabulary such that  $\{0,\max,R \}\subseteq \tau$, and $R$ a  $k$-ary relation symbol such that $R\notin \tau$. Then for every  $\tau\cup \{R\}$-sentence $\psi$ of $\ESOHORN$, in which $R$ appears only negatively,  there is a $\tau$-formula $\phi$ of $\SDHORN$ with free variables $z_1,\dotsc, z_k$  such that,  for all  $\mA$ and  $X\neq \emptyset$ with domain $\{z_1,\dotsc, z_k\}$:
\begin{equation}\label{neweq}
 \mA\models_X \phi \iff (\mA,\rel(X))\models\psi .
\end{equation}
\end{thm}
\begin{proof} 
Suppose $\psi\in \ESOHORN[\tau\cup \{R\}]$ is of the form
\[
\exists P_1 \dots \exists P_k  \forall \overline y \bigwedge_i C_i.
\]
It is easy to check that $\psi$ is logically equivalent to $\psi'$
\[\psi':= \exists R'(\forall \overline{x}(\neg R(\overline{x})\vee R'(\overline{x}))\wedge \psi [R'/R]),  \]
  where in $\psi [R'/R]$ all occurrences of $R$ are replaced by $R'$. Now $\psi'$ can  be transformed into  
\begin{equation}\label{quasi}
\varphi^*= \exists R'\exists P_1 \dots \exists P_k \forall  \overline x\forall \overline y (\bigwedge_i C_i[R'/R]\wedge  (\neg R(\overline{x})\vee R'(\overline{x}))).
\end{equation}
By the assumption that $R$ has only negative occurrences in $\psi$ it follows that the sentence \eqref{quasi} is also in $\ESOHORN$. Now we define the formula $\phi (z_1,\ldots,z_k)$ by first translating the sentence \eqref{quasi} into $\SDHORN$ as in Lemma \ref{SOHORN_in_SDHORN}, and then  replacing the subformula $\neg R(\overline{x})$ by
$\vee _{1\le i\le k} \neg z_i=x_i$. We note first that $\phi\in \SDHORN$. Furthermore, since the way $\phi$ is obtained from sentence \eqref{quasi} is the same as the translation given in \cite{kova09}, the formula  $\phi (z_1,\ldots,z_k)$ is as wanted.
\end{proof}

We say that a class $D$ of structures is downwards closed with respect to $R$ if for every structure $\mA$ the following holds: 
if $\mA\in D$ and 
$Q\subseteq R^{\mA}$ then $\mA'\in D$, where $\mA'$ arises from $\mA$ by replacing  $R^{\mA}$ by $Q$. As far as we know, it is an open question whether the $R$-negative fragment of $\ESOHORN$ can define all $\PTIME$-properties of successor structures which are downwards closed with respect to $R$. Note that for $\ESO$ and $\NP$ the analogous result holds, and hence the open formulas of dependence logic correspond exactly to the downwards closed $\NP$ properties.

\section{Conclusion}

Inspired by Gr\"adel's second-order characterization of \PTIME, we have studied two fragments of dependence logic \D. While the first restriction to Horn formulae still yields the full power of \NP, the more restricted type of Horn formulae captures the class \PTIME. Furthermore we showed that 
the open $\SDHORN$-formulae correspond exactly to the negative fragment of $\ESOHORN$. 



We conclude with the following open questions:
\begin{enumerate}
\item We showed that the model checking problem for the
fragment of $\SDHORN$ with $k$ universal quantifiers can be solved in TIME$(n^{2k})$ (see Corollary \ref{k_variables}). It would be interesting to prove a converse inclusion, i.e., that TIME$(n^{k})$ can be captured by $\SDHORN$ with $k'$ universal quantifiers, for some constant $k'>k$.
\item It would be interesting to know whether the $R$-negative fragment of $\ESOHORN$ (and the open $\SDHORN$ formulae) capture exactly all $\PTIME$ properties of successor structures which are downward closed with respect to $R$.

\end{enumerate}




\bibliographystyle{plain}
\bibliography{thi-hannover,Juha}


\appendix
\end{document}